\documentclass[letterpaper]{article} 
\usepackage{aaai2026}  
\usepackage{times}  
\usepackage{helvet}  
\usepackage{courier}  
\usepackage[hyphens]{url}  
\usepackage{graphicx} 
\usepackage{natbib}  
\usepackage{caption} 
\usepackage{algorithm}
\usepackage{algorithmic}

\usepackage{newfloat}
\usepackage{listings}
\DeclareCaptionStyle{ruled}{labelfont=normalfont,labelsep=colon,strut=off} 
\floatstyle{ruled}
\newfloat{listing}{tb}{lst}{}
\floatname{listing}{Listing}
\usepackage{soul} 
\usepackage[utf8]{inputenc} 
\usepackage{amsmath} 
\usepackage{amsthm, amssymb} 
\usepackage{booktabs} 
\usepackage{tikz} 
\usepackage{pgf} 
\usetikzlibrary{patterns,automata,arrows,shapes,topaths,trees,backgrounds,positioning,through,calc,arrows}
\definecolor{medgreen}{rgb}{0.0, 0.75, 0.0}

\usepackage{comment}

\usepackage[switch]{lineno}

\theoremstyle{definition}
\newtheorem{theorem}{Theorem}
\newtheorem{proposition}[theorem]{Proposition}
\newtheorem{lemma}[theorem]{Lemma}
\newtheorem{fact}[theorem]{Fact}
\newtheorem{example}[theorem]{Example}
\newtheorem{conjecture}{Conjecture}
\newtheorem{corollary}[theorem]{Corollary}
\newtheorem{definition}[theorem]{Definition}

\title{Stable Voting and the Splitting of Cycles}
\author{
    Wesley H. Holliday\textsuperscript{\rm 1},
    Milan Mossé\textsuperscript{\rm 1},
    Chase Norman\textsuperscript{\rm 2},
    Eric Pacuit\textsuperscript{\rm 3},
    Cynthia Wang\textsuperscript{\rm 2}
}
\affiliations{
    \textsuperscript{\rm 1}University of California, Berkeley\\
    \textsuperscript{\rm 2}Carnegie Mellon University\\
    \textsuperscript{\rm 3}University of Maryland\\
    wesholliday@berkeley.edu, milan\_mosse@berkeley.edu, chasen@cmu.edu, epacuit@umd.edu, cynthia4@andrew.cmu.edu

}

\begin{document}

\maketitle

\begin{abstract}
Algorithms for resolving majority cycles in preference aggregation have been studied extensively in computational social choice. Several sophisticated cycle-resolving methods, including Tideman's Ranked Pairs, Schulze's Beat Path, and Heitzig's River, are refinements of the Split Cycle (SC) method that resolves majority cycles by discarding the weakest majority victories in each cycle. Recently, Holliday and Pacuit proposed a new refinement of Split Cycle, dubbed Stable Voting, and a simplification thereof, called Simple Stable Voting (SSV). They conjectured that SSV is a refinement of SC whenever no two majority victories are of  the same size. In this paper, we prove the conjecture up to 6 alternatives and refute it for more than 6 alternatives. While our proof of the conjecture for up to 5 alternatives uses traditional mathematical reasoning, our 6-alternative proof and 7-alternative counterexample were obtained with the use of SAT solving. The SAT encoding underlying this proof and counterexample is applicable far beyond SC and SSV: it can be used to test properties of any voting method whose choice of winners depends only on the ordering of margins of victory by size.
\end{abstract}


\section{Introduction}\label{sec:intro}

Aggregating the preferences of multiple agents is a fundamental problem in group decision making, studied extensively in computational social choice \cite{Brandt2016Handbook}. One reason the problem is difficult is that when there are more than two alternatives, majority aggregation of ordinal preferences can lead to a \textit{majority cycle}, as occurs when, e.g., a majority of agents prefer $a$ to $b$, a majority prefer $b$ to $c$, and a majority prefer $c$ to $a$. The question of how to resolve such majority cycles has occupied social choice theorists since Condorcet~\shortcite{Condorcet1785}. One natural approach, already anticipated by Condorcet (see \citealt{Young1988}), is to resolve majority cycles based on the \textit{sizes} of majority victories within each majority cycle. If more agents rank alternative $a$ above $b$ than vice versa, then $a$ beats $b$ head-to-head---but by how much? How many \textit{more} agents rank $a$ above $b$ than rank $b$ above $a$? This is $a$'s \textit{margin of victory} over $b$, and these margins can help to resolve majority cycles. 

For example, the Split Cycle rule \cite{HP2023b} breaks each cycle at its weakest link. E.g., if $a$ beats $b$ head-to-head in a cycle (e.g., $a$ beats $b$, who beats $c$, who beats $d$, who beats $a$), but $a$'s margin of victory over $b$ is the smallest margin of victory in that cycle, then $a$ does not count as \textit{defeating} $b$, according to Split Cycle. After splitting all cycles at their weakest links in this way, there are always some undefeated alternatives, so this is one way of resolving majority cycles. See Figure \ref{SCexample} for an~example.

\begin{figure}[h!]

\centering

\begin{minipage}{1.5in}
    \begin{tikzpicture}
      \node[circle,draw,minimum width=0.25in] at (0,0) (a) {$a$}; 
      \node[circle,draw,minimum width=0.25in] at (3,0) (c) {$c$}; 
      \node[circle,draw,minimum width=0.25in] at (1.5,1.5) (b) {$b$}; 
      \node[circle,draw,minimum width=0.25in] at (1.5,-1.5) (d) {$d$}; 

      \path[->,draw,very thick,blue] (a) to node[fill=white,near end] {$10$} (c);
      \path[->,draw,thick] (d) to node[fill=white, near end] {$6$} (b);
      \path[->,draw,very thick,blue] (b) to node[fill=white] {$12$} (a);
      \path[->,draw,very thick,blue, dashed] (c) to node[fill=white] {$8$} (b);
      \path[->,draw,thick] (c) to node[fill=white] {$4$} (d);
      \path[->,draw,thick] (a) to node[fill=white] {$2$} (d);
    \end{tikzpicture}\vspace{.1in}
  \end{minipage}\hspace{.1in}\begin{minipage}{1.5in}\begin{tikzpicture}
      \node[circle,draw, minimum width=0.25in] at (0,0) (a) {$a$}; 
      \node[circle,draw,minimum width=0.25in] at (3,0) (c) {$c$}; 
      \node[circle,draw,minimum width=0.25in] at (1.5,1.5) (b) {$b$}; 
      \node[circle,draw,minimum width=0.25in] at (1.5,-1.5) (d) {$d$}; 

      \path[->,draw,thick] (a) to node[fill=white,near end] {$10$} (c);
      \path[->,draw,very thick,medgreen] (d) to node[fill=white,near end] {$6$} (b);
      \path[->,draw,very thick,medgreen] (b) to node[fill=white] {$12$} (a);
      \path[->,draw,thick] (c) to node[fill=white] {$8$} (b);
      \path[->,draw,thick] (c) to node[fill=white] {$4$} (d);
      \path[->,draw,very thick,medgreen, dashed] (a) to node[fill=white] {$2$} (d);

    \end{tikzpicture}\vspace{.1in}\end{minipage}
  
  \begin{minipage}{1.5in}
    \begin{tikzpicture}
      \node[circle,draw, minimum width=0.25in] at (0,0) (a) {$a$}; 
      \node[circle,draw,minimum width=0.25in] at (3,0) (c) {$c$}; 
      \node[circle,draw,minimum width=0.25in] at (1.5,1.5) (b) {$b$}; 
      \node[circle,draw,minimum width=0.25in] at (1.5,-1.5) (d) {$d$}; 

      \path[->,draw,very thick,red] (a) to node[fill=white,near end] {$10$} (c);
      \path[->,draw,very thick,red] (d) to node[fill=white,near end] {$6$} (b);
      \path[->,draw,very thick,red] (b) to node[fill=white] {$12$} (a);
      \path[->,draw,thick] (c) to node[fill=white] {$8$} (b);
      \path[->,draw,very thick,red,dashed] (c) to node[fill=white] {$4$} (d);
      \path[->,draw,thick] (a) to node[fill=white] {$2$} (d);
    \end{tikzpicture} 
  \end{minipage}\hspace{.1in}\begin{minipage}{1.5in}
    \begin{tikzpicture}[baseline=(current bounding box.center)]
    \node[circle,draw, minimum width=0.25in] at (0,0) (a) {$a$}; 
    \node[circle,draw,minimum width=0.25in] at (3,0) (c) {$c$}; 
    \node[circle,draw,minimum width=0.25in] at (1.5,1.5) (b) {$b$}; 
    \node[circle,draw,minimum width=0.25in, fill=blue!25] at (1.5,-1.5) (d) {$d$}; 

    \path[->,draw,thick] (a) to (c);
    \path[->,draw,thick] (d) to (b);
    \path[->,draw,thick] (b) to node  {} (a);
    \node  at (1.5,.5)  {}; 
    \node  at (2,0)  {}; 
  \end{tikzpicture}
  \end{minipage}

  \caption{Upper row and bottom left: three majority cycles, with their weakest links shown as dashed arrows. Bottom right: the head-to-head victories that count as \textit{defeats} according to Split Cycle. The only undefeated candidate is $d$.}\label{SCexample}
  \end{figure}
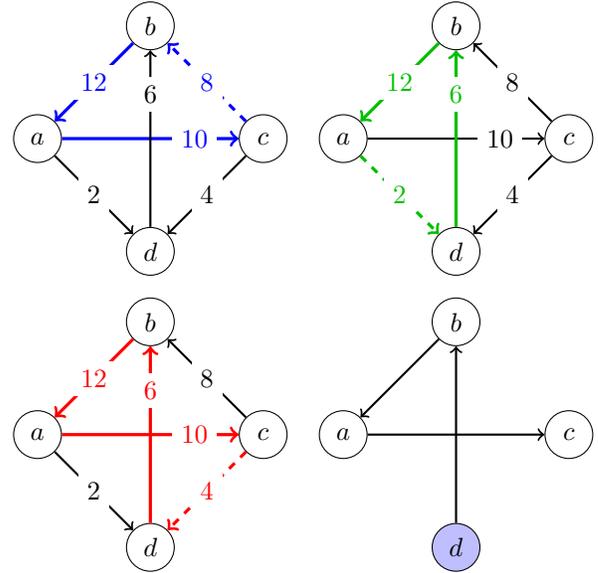

Split Cycle enjoys strong axiomatic motivation, e.g., by mitigating spoiler effects and avoiding the Strong No-Show Paradox \citep{ding2022axiomatic}. But this cannot be the end of the story, because there may be \textit{multiple} undefeated alternatives after splitting all cycles at their weakest links, as in  Figure \ref{TwoWinners}. How then do we choose from among the undefeated alternatives? Several sophisticated Condorcet voting methods,\footnote{A \textit{Condorcet method} is a voting method such that whenever there is a \textit{Condorcet winner}, an alternative that beats every other alternative head-to-head, the voting method elects that alternative.} including Tideman's \shortcite{Tideman1987} Ranked Pairs, Schulze's \shortcite{Schulze2011} Beat Path, and Heitzig's River (\citealt{Doring2025}), always choose  from among the undefeated alternatives according to Split Cycle; these voting methods \textit{refine} Split Cycle. 

\begin{figure}
\centering
\begin{minipage}{1.5in}\begin{tikzpicture}
\node[circle,draw,minimum width=0.25in] at (0,0)      (a) {$a$}; 
\node[circle,draw,minimum width=0.25in] at (3,0)      (b) {$c$}; 
\node[circle,draw,minimum width=0.25in] at (1.5,1.5)  (c) {$b$}; 
\node[circle,draw,minimum width=0.25in] at (1.5,-1.5) (d) {$d$};
\path[->,draw,thick] (a) to[pos=.7] node[fill=white] {$10$} (b);
\path[->,draw,thick] (c) to node[fill=white] {$2$} (a);
\path[->,draw,thick] (d) to node[fill=white] {$12$} (a);
\path[->,draw,thick] (c) to node[fill=white] {$8$} (b);
\path[->,draw,thick] (b) to node[fill=white] {$6$} (d);
\path[->,draw,thick] (d) to[pos=.7]  node[fill=white] {$4$} (c);
\end{tikzpicture}
\end{minipage}\qquad\begin{minipage}{1.5in}\begin{tikzpicture}
\node[circle,draw,minimum width=0.25in] at (0,0)      (a) {$a$}; 
\node[circle,draw,minimum width=0.25in] at (3,0)      (b) {$c$}; 
\node[circle,draw,minimum width=0.25in, fill=blue!25] at (1.5,1.5)  (c) {$b$}; 
\node[circle,draw,minimum width=0.25in, fill=blue!25] at (1.5,-1.5) (d) {$d$};
\path[->,draw,thick] (a) to  node  { } (b);

\path[->,draw,thick] (d) to node { } (a);
\path[->,draw,thick] (c) to node  { } (b);

\end{tikzpicture}
\end{minipage}
\caption{Left: margins between alternatives. Right: the defeats according to Split Cycle, yielding multiple undefeated candidates.}\label{TwoWinners}

\end{figure}
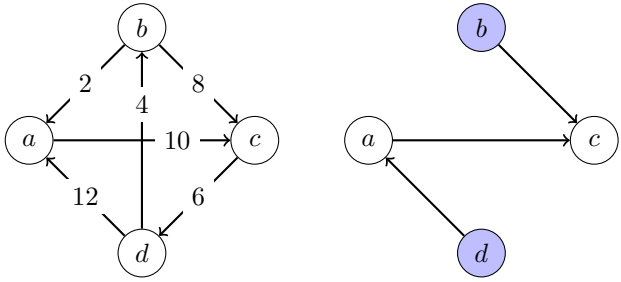

Recently, Holliday and Pacuit \shortcite{HP2023b} proposed a new refinement of Split Cycle, dubbed Stable Voting, and a simplification thereof, called Simple Stable Voting. The winning alternatives according to Simple Stable Voting (SSV) are defined \textit{recursively} as follows, for any profile $\mathbf{P}$ of agent preferences over a set $X$ of alternatives:
\begin{itemize}
\item if there is only one alternative in $X$, that alternative is the SSV winner in $\mathbf{P}$;
\item otherwise, list all head-to-head matches $a$ vs.~$b$ by decreasing margin of $a$ vs.~$b$,\footnote{This margin may be negative if $a$ loses head-to-head to $b$.} and find the first such match for which $a$ is the SSV winner in the profile $\mathbf{P}_{-b}$ with $b$ deleted; then $a$ is the SSV winner in $\mathbf{P}$.
\end{itemize}
For example, consider the head-to-head margins between alternatives shown in Figure \ref{TwoWinners}. If there are only three alternatives and no ties between them, then the SSV winner is the alternative with the smallest loss---or the alternative with no loss, if there is one \cite[Prop.~1]{HP2023b}. Using this Lemma, we can determine the SSV winner by going down the list of matches in Figure~\ref{TwoWinners} as follows:
\begin{itemize}
\item $d$ vs. $a$: margin of 12. But by the Lemma, $d$ \textit{loses} in the profile after deleting $a$ (see Figure~\ref{SSVexample}).
\item $a$ vs. $c$: margin of 10. But by the Lemma, $a$ \textit{loses} in the profile after deleting $c$ (see Figure~\ref{SSVexample}).
\item $b$ vs. $c$: margin of 8. But by the Lemma, $b$ \textit{loses} in the profile after deleting $c$ (see Figure~\ref{SSVexample}).
\item $c$ vs. $d$: margin of 6. But by the Lemma, $c$ \textit{loses} in the profile after deleting $d$ (see Figure~\ref{SSVexample}).
\item $d$ vs. $b$: margin of 4. And by the Lemma, $d$ \textit{wins} in the profile after deleting $b$ (see Figure~\ref{SSVexample}).
\end{itemize}
Thus, $d$ is the SSV winner, and SSV chooses from among the undefeated alternatives according to Split Cycle.

\begin{figure}[ht]
\centering
\begin{tikzpicture}[baseline=(current bounding box.center)]

    \node[circle,draw,minimum width=0.25in,fill=red!50] at (3,0) (c) {$c$}; 
    \node[circle,draw,minimum width=0.25in,fill=medgreen!50] at (1.5,1.5) (b) {$b$}; 
    \node[circle,draw,minimum width=0.25in, fill=red!50] at (1.5,-1.5) (d) {$d$};

    \path[->,draw,thick] (d) to (b);

    \path[<-,draw,thick] (c) to node[fill=white] {$8$} (b);
    \path[->,draw,thick] (c) to node[fill=white] {$6$} (d);

    \node[fill=white] at (1.5,.5)  {{$4$}}; 
    \node at (1.5,-2.5) {after deleting $a$};
  \end{tikzpicture}\vspace{.15in}\qquad\begin{tikzpicture}[baseline=(current bounding box.center)]
    \node[circle,draw, minimum width=0.25in, fill=red!50] at (0,0) (a) {$a$}; 
    \node[circle,draw,minimum width=0.25in, fill=red!50] at (3,0) (c) {$c$}; 
    \node[circle,draw,minimum width=0.25in, fill=medgreen!50] at (1.5,-1.5) (d) {$d$}; 

    \path[->,draw,thick] (a) to (c);

    \path[->,draw,thick] (c) to node[fill=white] {$6$} (d);
    \path[<-,draw,thick] (a) to node[fill=white] {$12$} (d);

    \node[fill=white] at (2,0)  {$10$};
      \node at (1.5,-2.5) {after deleting $b$};
  \end{tikzpicture}\qquad \begin{tikzpicture}[baseline=(current bounding box.center)]
    \node[circle,draw, minimum width=0.25in, fill=red!50] at (0,0) (a) {$a$}; 

    \node[circle,draw,minimum width=0.25in, fill=red!50] at (1.5,1.5) (b) {$b$}; 
    \node[circle,draw,minimum width=0.25in,fill=medgreen!50] at (1.5,-1.5) (d) {$d$};

    \path[->,draw,thick] (d) to (b);
    \path[->,draw,thick] (b) to node[fill=white] {$2$} (a);

    \path[<-,draw,thick] (a) to node[fill=white] {$12$} (d);

    \node[fill=white] at (1.5,.5)  {{$4$}}; 

      \node at (1.5,-2.5) {after deleting $c$};
  \end{tikzpicture}\qquad \begin{tikzpicture}[baseline=(current bounding box.center)]
    \node[circle,draw, minimum width=0.25in, fill=red!50] at (0,0) (a) {$a$}; 
    \node[circle,draw,minimum width=0.25in, fill=red!50] at (3,0) (c) {$c$}; 
    \node[circle,draw,minimum width=0.25in, fill=medgreen!50] at (1.5,1.5) (b) {$b$}; 

    \path[->,draw,thick] (a) to (c);

    \path[->,draw,thick] (b) to node[fill=white] {$2$} (a);
    \path[<-,draw,thick] (c) to node[fill=white] {$8$} (b);

    \node[fill=white] at (2,0)  {$10$}; 
      \node at (1.5,-1) {after deleting $d$};
  \end{tikzpicture}

\caption{Three-alternative elections considered in the computation of the SSV winners for the four-alternative election in Figure~\ref{TwoWinners}. SSV winners are shown in green.}\label{SSVexample}
\end{figure}
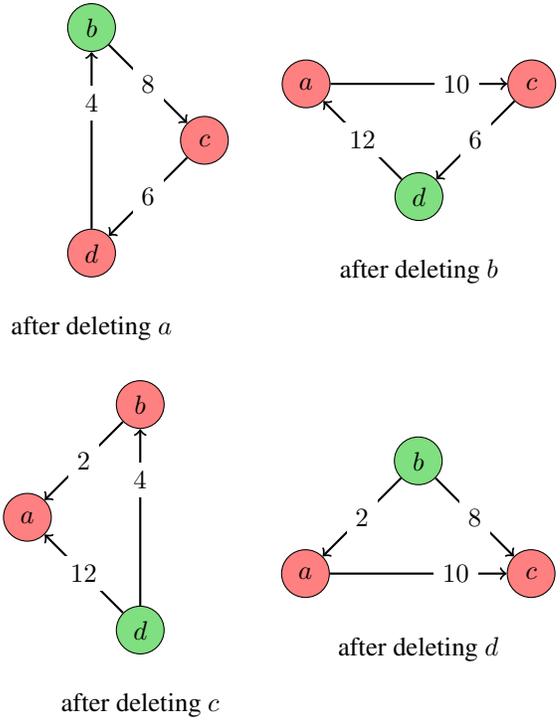

 The recursive definition of SSV seems to have little to do with Split Cycle:\footnote{However, the recursive definition of SSV may remind one of the more familiar Instant Runoff Voting (IRV), which can be defined recursively in a similar fashion: if there is only one alternative, they win; otherwise $a$ is an IRV winner if there is some alternative $b$ with the fewest first-place rankings from voters such that $a$ is the IRV winner \textit{after $b$ is deleted from all ballots}.} there is no mention of cycles, and SSV removes \textit{alternatives} from the election, rather than removing edges from the margin graph of the election. Yet it is remarkably difficult to find elections in which SSV does not choose from among the undefeated alternatives according to Split Cycle; extensive Monte Carlo searches and pen-and-paper work by multiple research groups yielded no counterexamples. Thus, the following conjecture remained unsolved.

 \begin{conjecture}[\citealt{HP2023}]\label{ConjInformal} SSV chooses from among the undefeated alternatives according to Split Cycle in every profile in which there are no tied margins.
 \end{conjecture}
\noindent As explained in Section~\ref{Prelims}, if true, this conjecture would imply that SSV is equivalent to the more complicated (to define) Stable Voting method that has been used in hundreds of elections \cite{HP2024data} at \text{https://stablevoting.org}. This would in turn imply that the simpler SSV definition could be the definition presented to voters, while the (empirically) faster Stable Voting implementation could be used behind the scenes (see \text{https://github.com/epacuit/stablevoting}).
 
  By ``no tied margins'' in Conjecture \ref{ConjInformal}, we mean there are no alternatives $a,b,c,d$ with $a\neq b$, $c\neq d$, and $(a,b)\neq (c,d)$ such that the margin of $a$ vs.~$b$ is equal to the margin of $c$ vs.~$d$. If we instead allow such ties, we must decide how to define Simple Stable Voting to handle them.
 
 Under one way of handling tied margins,  \textit{simultaneous elimination}, the version of the conjecture allowing tied margins was refuted by Yifeng Ding (personal communication) using an exhaustive computer search for counterexamples with up to 5 alternatives; we present one of his counterexamples with 5 alternatives in the Appendix. However, we also show in the Appendix that under another way of handling tied margins, \textit{parallel-universe tiebreaking}, those five-alternative counterexamples disappear, and there is a counterexample with tied margins only if there is one without tied margins. This returns us to the above conjecture, which has proven much more difficult to resolve. No counterexamples were revealed by an exhaustive search up to 6 alternatives or by extensive Monte Carlo sampling with more alternatives.

In this paper, we prove Conjecture \ref{ConjInformal} for up to 6 alternatives, and we refute it for more than 6 alternatives. While our proof of the conjecture for up to 5 alternatives (Section~\ref{5cands}) uses traditional mathematical methods, our 6-alternative proof and 7-alternative counterexample  were obtained with the use of satisfiability (SAT) solving on a computer (Section~\ref{section: SAT encoding}). This provides another example of the power of SAT solving in social choice theory (for precedents, see, e.g., \citealt{Brandt2014}, \citealt{Brandt2016}, \citealt{Brandt2017}, \citealt{Brandt2018}, \citealt{Brandt2022}, \citealt{Kluiving2020}, and \citealt{HNPZ2024}). Moreover, our SAT encoding used for Conjecture~\ref{ConjInformal} is applicable far beyond Split Cycle and SSV: it can be used to test properties of any voting method whose choice of  winners depends only on the ordering of margins of victory between alternatives by size (see Definition~\ref{OrdMarginMatrix}). 

All of the code associated with the paper is available in the repository at https://github.com/chasenorman/ssv-sc.

\section{Preliminaries}\label{Prelims}

To settle Conjecture \ref{ConjInformal} in the following sections, we need some preliminary setup.

First, note that Split Cycle and SSV, like Beat Path, Ranked Pairs, River, Minimax \cite{Simpson1969,Kramer1977}, Weighted Covering \cite{Dutta1999,Fernandez2018}, and all so-called C1 methods \cite{Fishburn1977}, do not  need the numerical values of the margins between alternatives in order to determine the winners; it suffices to have an ordering of pairs of alternatives by margin, as  in the following definition.

\begin{definition}\label{OrdMarginMatrix} An \textit{ordinal margin matrix} is a pair $\mathcal{M}={(X,\succ)}$ of a nonempty finite set $X$ (of alternatives) and a strict weak order (i.e., an asymmetric and negatively transitive relation) $\succ$ on $X^2$ satisfying \textit{ordinal skew-symmetry}:  for all $a,b,c,d\in X$, \[\mbox{$(a,b)\succ (c,d)$ if and only if $(d,c)\succ (b,a)$.}\]

We say $\mathcal{M}$ is \textit{linear} if the restriction of $\succ$ to $X^2\setminus \{(a,a)\mid a\in X\}$ is a linear order. If $(a,b)\succ (c,d)$, we say that $(a,b)$ is \textit{stronger} than $(c,d)$. If $(a,b)\succ (b,a)$, we say that $a$ has a \textit{positive margin over} $b$.

Given an ordinal margin matrix $(X,\succ)$, its \textit{majority graph} $G(X,\succ)$ is the directed graph whose set of nodes is $X$ with an edge from $a$ to $b$ if $(a,b)\succ (b,a)$.
\end{definition}

To relate Definition \ref{OrdMarginMatrix} to the standard setup in voting theory (see \citealt{Zwicker2016}), recall that a \textit{preference profile} is a function $\mathbf{P}:V\to \mathcal{L}(X)$ where $V$ is a nonempty set of voters, $X$ is a nonempty set of alternatives, and $\mathcal{L}(X)$ is the set of all linear orders on $X$.\footnote{Everything to follow holds if we allow strict weak orders instead of linear orders.} Given $a,b\in X$, the \textit{margin of $a$ vs.~$b$ in $\mathbf{P}$}, denoted $\mbox{Margin}_\mathbf{P}(a,b)$, is the number of voters who rank $a$ above $b$ minus the number who rank $b$ above $a$:
\[|\{i \in V\mid (a,b)\in \mathbf{P}(i)\}| - |\{i\in V\mid (b,a)\in \mathbf{P}(i)\}|.\]
Any profile $\mathbf{P}$ induces an ordinal margin matrix $\mathcal{M}(\mathbf{P})=(X,\succ)$ by setting $(a,b)\succ (c,d)$ iff $\mbox{Margin}_\mathbf{P}(a,b)>\mbox{Margin}_\mathbf{P}(c,d)$. Conversely, we have the following representation result, proved in the Appendix.

\begin{lemma}\label{RepThm}
For any ordinal margin matrix $\mathcal{M}$, there is a preference profile $\mathbf{P}$ such that  $\mathcal{M}(\mathbf{P})=\mathcal{M}$.
\end{lemma}

Let us now define Split Cycle on ordinal margin matrices.

\begin{definition}\label{SCdef} Given an ordinal margin matrix $\mathcal{M}=(X,\succ)$ and $a,b\in X$, we say that $a$ \textit{defeats} $b$ if
\begin{enumerate}
\item\label{SCdef1} $(a,b)\succ(b,a)$ and 
\item\label{SCdef2} there is no cycle in the majority graph $G(X,\succ)$ such that $(a,b)$ is an edge that is weakest\footnote{Possibly tied for weakest, if $\succ$ is not linear.} in the cycle according to $\succ$. 
\end{enumerate}
Then the \textit{Split Cycle winners} for $\mathcal{M}$ are the alternatives that are not defeated by any other alternative.\end{definition}

We recall the following from \citealt{HP2023}. 

\begin{lemma} For any ordinal margin matrix $\mathcal{M}=(X,\succ)$,  there is at least one Split Cycle winner.
\end{lemma}
\begin{proof} Suppose for contradiction that every alternative is defeated in the sense of Definition \ref{SCdef}. Then since the set of alternatives is finite, it follows that there is a cycle of defeats in $G(X,\succ)$. But this cycle must have some $\succ$-minimal edge, in which case by Definition \ref{SCdef}.\ref{SCdef2} that edge is not a defeat, so we have a contradiction.\end{proof}

To define SSV, we need to define the result of deleting an alternative.

\begin{definition} Given an ordinal margin matrix $\mathcal{M}={(X,\succ)}$ with $|X|>1$ and $b\in X$, define $\mathcal{M}_{-b}=(X_{-b}, \succ_{-b})$ where $X_{-b}={X\setminus\{b\}}$ and $\succ_{-b}$ is the restriction of $\succ$ to $(X_{-b})^2$.\end{definition}
\noindent Clearly $\mathcal{M}_{-b}$ is also an ordinal margin matrix.

\begin{definition}\label{def:SSV} For any linear ordinal margin matrix $\mathcal{M}$, the Simple Stable Voting (SSV) winners are defined recursively:
\begin{enumerate}
\item if $\mathcal{M}$ has only one alternative, this alternative is the SSV winner in $\mathcal{M}$;
\item if $\mathcal{M}$ has more than one alternative, then the SSV winner in $\mathcal{M}$ is the first coordinate of the $\succ$-maximal pair of alternatives $(a,b)$ such that $a$ is the SSV winner in $\mathcal{M}_{-b}$.
\end{enumerate} 
\end{definition}

We recall the following from \citealt{HP2023b}.

\begin{lemma}\label{UniqueSSV} For every linear ordinal margin matrix $\mathcal{M}=(X,\succ)$, there is a unique SSV winner.
\end{lemma}
\begin{proof} By induction on $|X|$. The base case where $|X|=1$ is immediate. Suppose $|X|=n+1$. For any $c\in X$, there is an SSV winner in $\mathcal{M}_{-c}$ by the inductive hypothesis. Then since $\succ$ is a linear order, there is a unique $\succ$-maximal pair $(a,b)$ such that $a$ is the SSV winner in $\mathcal{M}_{-b}$, so $a$ is the SSV winner in~$\mathcal{M}$.\end{proof}

We can now state Conjecture \ref{ConjInformal} more formally.

\setcounter{conjecture}{0}
\begin{conjecture}[Conjecture~3.21 of \citealt{HP2023}]\label{Conj} For every linear ordinal margin matrix, the SSV winner is among the Split Cycle winners.
\end{conjecture}

This conjecture, if true, would also settle the relation between SSV and what Holliday and Pacuit \shortcite{HP2023} call the Stable Voting (SV) method. SV is defined just like SSV, except that in part 2 of Definition \ref{def:SSV}, the definition of SV replaces ``such that $a$ is the SSV winner in $\mathcal{M}_{-b}$'' with ``such that $a$ is a Split Cycle winner in $\mathcal{M}$ and the SV winner in $\mathcal{M}_{-b}$,'' so SV refines Split Cycle by definition. 

\begin{lemma} For any positive integer $n$, if the restriction of Conjecture~\ref{Conj} up to $n$ alternatives holds, then SV is equivalent to SSV for linear ordinal margin matrices up to $n$ alternatives.\end{lemma}

\begin{proof} By induction on $n$, with a trivial base case. For induction, assume SV and SSV agree on $n$ alternative elections. Then for any $\mathcal{M}$ with $n+1$ alternatives, the SSV winner in $\mathcal{M}$, i.e., the first member of the $\succ$-maximal pair  $(a,b)$ such that $a$ is the SSV winner in $\mathcal{M}_{-b}$, is the same, assuming Conjecture~\ref{Conj} for $n+1$ alternatives, as the first member of the $\succ$-maximal pair  $(a,b)$ such that $a$ is a Split Cycle winner in $\mathcal{M}$ and the SSV winner in $\mathcal{M}_{-b}$, which is the same, by the inductive hypothesis, as the first member of the $\succ$-maximal pair $(a,b)$ such that $a$ is a Split Cycle winner in $\mathcal{M}$ and the SV winner in $\mathcal{M}_{-b}$, i.e., the SV winner in~$\mathcal{M}$.\end{proof}

\section{Proof of Conjecture for $\leq5$ Alternatives}\label{5cands}

Our first main result is a proof of Conjecture \ref{Conj} for up to 5 alternatives.

\begin{theorem}\label{prop:proof} For any linear ordinal margin matrix $\mathcal{M}$ with up to 5 alternatives, the SSV winner is a Split Cycle winner.
\end{theorem}
\begin{proof} By induction on the number of alternatives, with a trivial base case. Let $a$ be the SSV winner in $\mathcal{M}$ as witnessed by the match $a$ vs.~$b$, so $a$ is the SSV winner in $\mathcal{M}_{-b}$. Suppose for contradiction that $a$ is defeated according to Split Cycle (SC-defeated), say by $c$. We must have $c=b$. For if $c\neq b$, then $c$ still SC-defeats $a$ in $\mathcal{M}_{-b}$, so by the inductive hypothesis, $a$ is not the SSV winner in $\mathcal{M}_{-b}$, contradicting our assumption. Thus, we conclude that $b$ SC-defeats~$a$, which implies that $b$ has a positive margin over $a$.

The key idea of the proof is to recursively define
\begin{eqnarray*}
a_0 &=& a \\
a_{n+1} & = & \mbox{the SSV winner in }\mathcal{M}_{-a_n}.
\end{eqnarray*}
Since the match $a$ vs.~$b$ witnesses $a$ being the SSV winner, we claim that for all $n$,
\begin{equation}
(a_n,a_{n+1})\succ(b,a)\label{StrengthFact}
\end{equation}
and hence $a_n$ has a positive margin over $a_{n+1}$. Suppose for contradiction that the margin of $a_n$ over $a_{n+1}$ is smaller than that of $b$ over $a$. Then by ordinal skew-symmetry, the margin of $a_{n+1}$ over $a_n$ is larger than that of $a$ over $b$, in which case we consider the match $a_{n+1}$ vs.~$a_n$ before the match $a$ vs.~$b$ in step 2 of Definition~\ref{def:SSV}. Then because $a_{n+1}$ is the SSV winner in $\mathcal{M}_{-a_n}$, the match $a$ vs.~$b$ is not the first match in the list of matches by descending margin such that the first candidate wins after removing the second, contradicting our assumption at the beginning of the proof.

Since $\mathcal{M}$ is finite, the sequence $b,a_0,a_1,\dots$ must loop back onto itself. Since the shortest cycle is of length~3, there are only three cases to consider.

Case 1: the sequence loops back at $b$, as shown below:

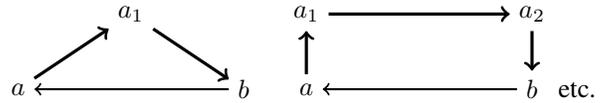
\begin{tikzpicture}

\node  at (0,0) (a) {$a$}; 
\node  at (3,0) (b) {$b$}; 
\node at (1.5,1) (c) {$a_1$}; 

\path[->,draw,very thick] (c) to node {} (b);
\path[->,draw,thick] (b) to node  {} (a);
\path[->,draw,very thick] (a) to node  {} (c);

\end{tikzpicture}\quad \begin{tikzpicture}

\node  at (0,0) (a) {$a$}; 
\node  at (3,0) (b) {$b$}; 
\node at (0,1) (c) {$a_1$}; 
\node at (3,1) (d) {$a_2$}; 

\path[->,draw,very thick] (c) to node {} (d);
\path[->,draw,very thick] (d) to node {} (b);
\path[->,draw,thick] (b) to node  {} (a);
\path[->,draw,very thick] (a) to node  {} (c);

\node at (3.6,0) (e) {etc.};

\end{tikzpicture}

\noindent Then we have a path from $a$ to $b$ in the graph such that the margin of each edge is greater than that of the edge from $b$ to $a$ by (\ref{StrengthFact}). But this contradicts the fact that $b$ SC-defeats $a$.

Case 2: the sequence loops back at $a$, as shown below:

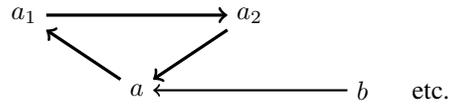
\begin{tikzpicture}

\node  at (0,0) (a) {$a$}; 
\node  at (3,0) (b) {$b$}; 
\node at (-1.5,1) (c) {$a_1$}; 
\node at (1.5,1) (d) {$a_2$}; 

\path[->,draw,very thick] (c) to node {} (d);
\path[->,draw,very thick] (d) to node {} (a);
\path[->,draw,thick] (b) to node  {} (a);
\path[->,draw,very thick] (a) to node  {} (c);

\end{tikzpicture}\quad \begin{tikzpicture} \node at (0,0) (a) {etc.};
\end{tikzpicture}

\noindent That is, $a$ is the SSV winner in $\mathcal{M}_{-a_k}$ for some $k>0$. Then by the inductive hypothesis, $a$ is SC-undefeated in $\mathcal{M}_{-a_k}$, so from the fact that $b$ has a positive margin over $a$ in $\mathcal{M}_{-a_k}$, it follows that this is the weakest margin in a cycle containing $a$ and $b$ in $\mathcal{M}_{-a_k}$ and hence in $\mathcal{M}$, contradicting that $b$ SC-defeats $a$~in~$\mathcal{M}$.

Case 3: the sequence loops back at $a_1$, as shown below:
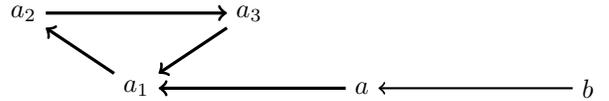
\begin{tikzpicture}

\node  at (0,0) (a) {$a$}; 
\node  at (3,0) (b) {$b$}; 
\node at (-3,0) (c) {$a_1$}; 
\node at (-4.5,1) (d) {$a_2$}; 
\node at (-1.5,1) (e) {$a_3$}; 

\path[->,draw,very thick] (c) to node {} (d);
\path[->,draw,very thick] (d) to node {} (e);
\path[->,draw,very thick] (e) to node {} (c);
\path[->,draw,thick] (b) to node  {} (a);
\path[->,draw,very thick] (a) to node  {} (c);

\end{tikzpicture}

\noindent In $\mathcal{M}_{-a_3}$, the SSV winner is $a_1$. Hence by the inductive hypothesis, $a_1$ is SC-undefeated in $\mathcal{M}_{-a_3}$, so there must be a path from $a_1$ back to $a$ whose edges have stronger margins than the margin of $a$ over $a_1$, which recall is stronger than the margin of $b$ over $a$. Hence we have at least the following relations, with thicker edges indicating stronger margins:\footnote{Equal edge thickness does not imply equal margins.}

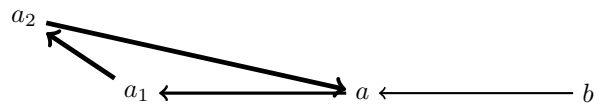
\begin{tikzpicture}

\node  at (3,0) (b) {$b$}; 
\node  at (0,0) (a) {$a$}; 
\node at (-3,0) (c) {$a_1$}; 
\node at (-4.5,1) (d) {$a_2$}; 

\path[->,draw,line width=1.8pt] (c) to node {} (d);
\path[->,draw,thick] (b) to node  {} (a);
\path[->,draw,very thick] (a) to node  {} (c);
\path[->,draw,line width=1.8pt] (d) to node  {} (a);

\end{tikzpicture}

\noindent Now since $a$ wins in $\mathcal{M}_{-b}$ and hence is SC-undefeated by the inductive hypothesis, the edge from $a_2$ to $a$ must be the weakest link in some cycle in $\mathcal{M}_{-b}$. Since the edge from $a_2$ to $a$ is not weakest in the cycle involving $a$, $a_1$, and $a_2$, it follows that there must be an edge from $a$ to $a_3$ that is stronger than the edge from $a_2$ to $a$ (here dotted lines indicate no claim about relative strength of margins):
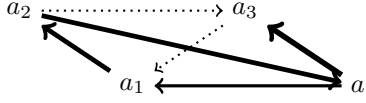
\begin{center}
\begin{tikzpicture}

\node  at (0,0) (a) {$a$}; 
\node at (-3,0) (c) {$a_1$}; 
\node at (-4.5,1) (d) {$a_2$}; 
\node at (-1.5,1) (e) {$a_3$}; 

\path[->,draw,line width=1.8pt] (c) to node {} (d);
\path[->,draw,dotted, thick] (d) to node {} (e);
\path[->,draw,dotted,thick] (e) to node {} (c);
\path[->,draw,very thick] (a) to node  {} (c);
\path[->,draw,line width=2.1pt] (a) to node  {} (e);
\path[->,draw,line width=1.8pt] (d) to node {} (a);

\end{tikzpicture}
\end{center}
Since $a_3$ wins in $\mathcal{M}_{-a_2}$ and hence is SC-undefeated by the inductive hypothesis, the edge from $a$ to $a_3$ must be the weakest in some cycle in $\mathcal{M}_{-a_2}$. But by inspection of the initial graph in Case 3, such a cycle must contain the edge from $b$ to $a$, which is weaker than the edge from $a$ to $a_1$ by (\ref{StrengthFact}), which is weaker than the edge from $a_2$ to $a$, which is weaker than the edge from $a$ to $a_3$---a contradiction.\end{proof}

\noindent Our proof does not generalize. It finds cycles $C_1,...,C_n$ with increasingly strong edges, ending with an impossible $C_n$. $C_1$ always exists, but $\mathcal{M}$ may have exponentially many cycles, and some $C_i$ may block induction by including none of the alternatives that win when either of two others is removed.

\section{SAT Encoding}\label{section: SAT encoding}

In this section, we prove Conjecture \ref{ConjInformal} for 6 alternatives and refute it for 7 with the help of SAT solving. First, we recall the necessary basic notions from propositional logic (see, e.g., \citealt[Ch.~1]{Enderton2001}). We begin with the formal language and semantics of propositional logic.

\begin{definition} Given any set $\{x_1,\dots,x_k\}$ of \textit{propositional variables}, the set $L(x_1,\dots,x_k)$ is the smallest set of strings over $\{x_1,\dots,x_k,(,),\neg,\wedge,\vee,\to\}$ such that:
\begin{enumerate}
\item each $x_i$ is in $L(x_1,\dots,x_k)$;
\item if $\varphi$ and $\psi$ are in $L(x_1,\dots,x_k)$, so are $\neg\varphi$,  $(\varphi\wedge\psi)$, $(\varphi\vee\psi)$, and $(\varphi\to\psi)$.
\end{enumerate}
We call elements of $L(x_1,\dots,x_k)$ the \textit{Boolean formulas} that can be constructed from  $x_1,\dots,x_k$. Examples include ${x_1\to\neg x_2}$, $x_3\wedge (x_2\vee x_4)$, etc.~(omitting outermost parentheses since no ambiguity will~arise).
\end{definition}
\begin{definition}
    A \textit{truth-value assignment}  $v$ for $\{x_1,\dots,x_k\}$ is a function mapping each variable $x_i$  to either \textsc{True} or \textsc{False}. Such an assignment extends recursively to an arbitrary formula $\varphi$ in $L(x_1,\dots,x_k)$  as follows:
    \begin{align*}
    \hat{v}(x_i) & = v(x_i) \\
            \hat{v}(\neg\varphi ) &= \textsc{True} \text{ iff } \hat{v}(\varphi) = \textsc{False};\\
        \hat{v}(\varphi \lor \psi) &= \textsc{True} \text{ iff }\hat{v}(\varphi) = \textsc{True} \text{ or } \hat{v}(\psi) = \textsc{True};\\
        \hat{v}(\varphi \land \psi) &= \textsc{True} \text{ iff }\hat{v}(\varphi) = \textsc{True} \text{ and } \hat{v}(\psi) = \textsc{True};\\
            \hat{v}(\varphi \to \psi) &= \textsc{True} \text{ iff }\hat{v}(\varphi) = \textsc{False} \text{ or } \hat{v}(\psi) = \textsc{True}.
    \end{align*}
\end{definition}

The Boolean satisfiability (SAT) problem  asks whether for a given Boolean formula $\varphi$, there exists a truth-value assignment $v$ for the variables in $\varphi$ such that $\hat{v}(\varphi)=\textsc{True}$. If so, we say $\varphi$ is \emph{satisfiable}; otherwise $\varphi$ is \emph{unsatisfiable}. 

SAT solvers are programs that accept a Boolean formula $\varphi$ written in Conjunctive Normal Form\footnote{I.e., a formula $(l_{1,1}\vee\dots\vee l_{1,n_1})\wedge \dots\wedge (l_{m,1}\vee\dots\vee l_{m,n_m})$ where each $l_{p,q}$ is of the form $x_i$ or $\neg x_i$.} (CNF) and either produce an assignment for its variables such that $\varphi$ evaluates to \textsc{True} or produce a proof that the formula is unsatisfiable. A formula that is not in CNF can be converted to an equisatisfiable CNF formula (that is, a formula that is satisfiable if and only if the original formula is satisfiable) with only a constant-factor increase in size by way of a Tseitin  transformation \cite{Tseitin}, which introduces \emph{auxiliary} variables that were not present in the original formula. 

In this section, we construct a formula for any $n$ that is satisfiable if and only if there exists a linear ordinal margin matrix for $n$ alternatives where Simple Stable Voting does not select a Split Cycle winner, contrary to Conjecture \ref{Conj}. We will write the formula as a conjunction of constraints in propositional logic, assuming a Tseitin transformation is performed to obtain a CNF formula for the SAT solver. 

\subsection{Linear Ordinal Margin Matrices}\label{section: ordinal margin matrices}

\newcommand{\gt}[4]{s_{(#1, #2), (#3, #4)}}

We begin by encoding the property of being a linear  ordinal margin matrix over $n$ alternatives. Without loss of generality we assume the alternative set $X$ is $[n]$, the set of numbers from $1$ to $n$. We begin by introducing variables $\gt{a}{b}{c}{d}$ for all $a, b, c, d \in [n]$ with $a\neq b$, $c\neq d$, and $(a,b)\neq (c,d)$, representing $(a, b) \succ (c, d)$ in the  ordinal margin matrix. For all such variables, we add the following constraints to our formula (where $e\neq f$, $(e,f)\neq (a,b)$, and $(e,f)\neq (c,d)$):
\begin{itemize}
\item Transitivity: $(\gt{a}{b}{c}{d} \land \gt{c}{d}{e}{f}) \to \gt{a}{b}{e}{f}$;
\item Asymmetry: $\gt{a}{b}{c}{d} \to \lnot \gt{c}{d}{a}{b}$;
\item Connectedness: $\lnot \gt{c}{d}{a}{b} \to \gt{a}{b}{c}{d}$; 
\item Ordinal skew-symmetry: $\gt{a}{b}{c}{d} \leftrightarrow \gt{d}{c}{b}{a}$.\footnote{Enforced by mapping $\gt{a}{b}{c}{d},\gt{d}{c}{b}{a}$ to one variable.}
\end{itemize}

\begin{fact}\label{ValToMat} Given an assignment $v$, define $\mathcal{M}_v=([n],\succ)$ by $(a,b)\succ (c,d)$ iff (i) $s_{(a,b),(c,d)}$ exists and $v(s_{(a,b),(c,d)})=\textsc{True}$ or (ii) $a=b$, $c\neq d$, and $v(s_{(d,c),(c,d)})=\textsc{True}$ or (iii) $a\neq b$, $c=d$, and $v(s_{(a,b),(b,a)})=\textsc{True}$. If $v$ satisfies the constraints above, then $\mathcal{M}_v$ is an ordinal margin matrix.
\end{fact}

\subsection{Simple Stable Voting}
The SSV winner in a linear  ordinal margin matrix $\mathcal{M}$ on the set $[n]$ is defined in terms of the SSV winners in restrictions of $\mathcal{M}$ to nonempty sets $V \subseteq [n]$. To encode Definition \ref{def:SSV}, we introduce variables $SSV_{V, a}$ for each $V \subseteq [n]$ and $a \in V$ representing that $a$ is the SSV winner in the restriction of $\mathcal{M}$ to $V$. We add two constraints for each nonempty $V \subseteq [n]$:
\begin{itemize}
\item at least one winner: $\underset{a \in V}{\bigvee} SSV_{V, a}$;
\item at most one winner: $\underset{\substack{a,b \in V\\a\ne b}}{\bigwedge} (\lnot SSV_{V, a} \lor \lnot SSV_{V, b})$.
\end{itemize}
Note that the ``at least one'' condition subsumes the base case of Definition \ref{def:SSV}. For the second part of the definition, we add the following constraint for all nonempty $V \subseteq [n]$ and  ${a, b \in V}$ such that $a \ne b$:

\begin{equation}
\scriptsize
(SSV_{V\setminus\{b\}, a} \land \bigwedge_{\substack{c, d \in V\\c \ne d\\ (a, b) \ne (c, d)}}(SSV_{V\setminus\{d\}, c} \to \gt{a}{b}{c}{d} )) \to SSV_{V, a}.\label{SSVconstraint}
\end{equation}

\noindent Thus, if $SSV_{V\setminus\{b\}, a}$ and $(a, b) \succ (c, d)$ for any $c \ne d$ satisfying $SSV_{V\setminus \{d\}, c}$, then we must have $SSV_{V, a}$. Because we have ensured that there is at most one SSV winner in $V$, this will be the only winner, completing our encoding of SSV. 

\begin{fact}\label{ValToMat2} If an assignment $v$ satisfies all of the constraints above, then $v(SSV_{V,a})=\textsc{True}$ if and only if $a$ is the SSV winner in the restriction of $\mathcal{M}_v$ to $V$.\end{fact}

\begin{proof} By induction on $|V|$. For $V=\{a\}$,  $v(SSV_{V,a})=\textsc{True}$ by the ``at least one'' constraint. Suppose $|V|>1$. Let $x$ be the unique SSV winner, given by Lemma \ref{UniqueSSV}, in the restriction of $\mathcal{M}_v$ to $V$, which we will call $\mathcal{N}$. Since $x$ is the SSV winner in $\mathcal{N}$, there is a $y\in V$ such that $x$ is the SSV winner in $\mathcal{N}_{-y}$, and $(x,y)\succ (s,t)$ for any  $(s,t)\neq (x,y)$ such that $s$ is the SSV winner in $\mathcal{N}_{-t}$. Then by the inductive hypothesis, $v(SSV_{V\setminus\{y\}, x})=\textsc{True}$, and for any  $(s,t)\neq (x,y)$ such that $v(SSV_{V\setminus\{t\}, s})=\textsc{True}$, we have $(x,y)\succ (s,t)$ and hence $v(s_{(x,y),(s,t)})=\textsc{True}$. Since $v$ satisfies (\ref{SSVconstraint}), it follows that $v(SSV_{V,x})=\textsc{True}$. Then by the ``at most one'' constraint, $v(SSV_{V,a})=\textsc{True}$ iff $a=x$.\end{proof}

\subsection{Split Cycle}

An alternative is an SC winner iff it is not SC-defeated. In a linear ordinal margin matrix $\mathcal{M}$, an alternative $a$ SC-defeats $b$ iff $a$ has a positive margin over $b$ and $a$ is not reachable from $b$ in $G(\mathcal{M})$ using only edges stronger than $(a, b)$.

For each $(a, b)$ with $a\neq b$ and each $c$, we create a variable $r_{(a, b), c}$ and impose the following constraints for all $a, b, c, d \in [n]$ such that $a \ne b$, $c \ne d$, and $(c, d) \ne (a, b)$:
\begin{align}
     &r_{(a,b),b}; \notag \\
     &(r_{(a, b), c} \wedge s_{(c,d),(d,c)}\land s_{(c, d), (a,b)}) \rightarrow r_{(a,b),d}.\label{ReachImp}
\end{align}
\begin{fact}\label{ReachFact} Let $v$ be an assignment satisfying all the constraints above. If $d$ is reachable from $b$ in $G(\mathcal{M}_v)$ using only edges stronger than $(a,b)$ in $\mathcal{M}_v$, then $v(r_{(a,b),d})=\textsc{True}$.\footnote{We do not need an `if and only if' here for our purposes.}
\end{fact}
\begin{proof} We prove by induction on $k\in\mathbb{N}$ that if $d$ is reachable from $b$ using a path of length $k$ whose edges are stronger than $(a,b)$ in $\mathcal{M}_v$, then $v(r_{(a,b),d})=\textsc{True}$. For $k=0$, one of our constraints is that $v(r_{(a,b),b})=\textsc{True}$. Now assuming the claim holds for paths of length $k$, suppose that $d$ is reachable from $b$ using a path of length $k+1$ whose edges are stronger than $(a,b)$ in $\mathcal{M}_v$. Call the last edge in the path $(c,d)$. Then $c$ is reachable from $b$ using a path of length $k$ whose edges are stronger than $(a,b)$, so by the inductive hypothesis, $v(r_{(a,b),c})=\textsc{True}$, and by definition of $\mathcal{M}_v$, $(c,d)\succ (d,c)$ and $(c,d)\succ (a,b)$ together imply $\hat{v}(s_{(c,d),(d,c)}\wedge s_{(c,d),(a,b)})=\textsc{True}$. Then since $v$ satisfies (\ref{ReachImp}),  $v(r_{(a,b),d})=\textsc{True}$, which completes the proof.\end{proof}
Next, for each $b\in [n]$, we introduce a variable $SC_b$ subject to the following constraint: 
\begin{equation}\big(\bigwedge_{\substack{a \in [n] \\ a \ne b}} (s_{(a,b),(b,a)}\to r_{(a, b), a})\big) \rightarrow SC_{b}.\label{SCeq}\end{equation}

\noindent The following is straightforward to check.

\begin{fact}\label{MatToVal} Given an ordinal margin matrix $\mathcal{M}=(X,\succ)$, let $v_\mathcal{M}$ be an assignment such that for all $s_{(a,b),(c,d)}$:
\begin{enumerate}
\item $v_\mathcal{M}(s_{(a,b),(c,d)})=\textsc{True}$ iff $(a,b)\succ (c,d)$;
\item $v_\mathcal{M}(SSV_{V,a})=\textsc{True}$ iff $a$ is the SSV winner in the restriction of $\mathcal{M}$ to $V$;
\item $v_\mathcal{M}(r_{(a,b),c})=\textsc{True}$ iff $b=c$ or $c$ can be reached from $b$ in the majority graph $G(\mathcal{M})$ using only edges stronger than $(a,b)$ according to $\succ$;
\item $v_\mathcal{M}(SC_a)=\textsc{True}$ iff $a$ is an SC winner in $\mathcal{M}$.
\end{enumerate}
Then all of the constraints above are \textsc{True} according to $\hat{v}_\mathcal{M}$.
\end{fact}

We can now prove the main fact about our SAT encoding.

\begin{proposition}\label{EncodingProp} Where $\varphi$ is the conjunction of all constraints introduced above, there exists a linear ordinal margin matrix in which the SSV winner is not an SC winner iff ${SSV_{[n],1}\wedge\neg SC_1 \wedge\varphi}$ is satisfiable.
\end{proposition}
\begin{proof}
From left to right, suppose there is a linear ordinal margin matrix $\mathcal{M}$ in which the SSV winner, who we relabel as $1$, is not an SC winner. Then by Fact \ref{MatToVal}, $v_\mathcal{M}$ satisfies ${SSV_{[n],1}\wedge\neg SC_1 \wedge\varphi}$. From right to left, suppose there is an assignment $v$ satisfying $SSV_{[n],1}\wedge\neg SC_1 \wedge\varphi$. Then by Fact~\ref{ValToMat}, we can consider $\mathcal{M}_v$, in which $1$ is the SSV winner by Fact~\ref{ValToMat2}. Since $v$ satisfies $\neg SC_1$ and (\ref{SCeq}), there is some ${a\neq 1}$ with $v(s_{(a,1),(1,a)})=\textsc{True}$ and $v(r_{(a,1),a})=\textsc{False}$, which with Fact \ref{ReachFact} implies that $a$ SC-defeats $1$.\end{proof}

\subsection{Symmetry Breaking}
\label{subsec:symmetry}

With the constraints added thus far, a linear ordinal margin matrix over $n$ alternatives may be represented in $n!$ ways, one for each permutation of $[n]$. To simplify our encoding and make SAT solving more efficient, we \emph{symmetry break} on permutations of the alternatives: any given  ordinal margin matrix is represented uniquely up to isomorphism. We do so by asserting that for all ${i \in [n]}$ with $i\geq 2$, alternative 1 is the SSV winner in $[i]$ as witnessed by the pair~$(1, i)$:
\[ SSV_{[i-1], 1} \land \bigwedge_{\substack{c, d \in [i]\\c \ne d\\ (1, i) \ne (c, d)}}\left(SSV_{[i]\setminus\{d\}, c} \to \gt{1}{i}{c}{d} \right).\]
\noindent Adding these constraints to $\varphi$ in Proposition \ref{EncodingProp}, an analogous proposition holds. If $a \in X$ is the SSV winner for a linear ordinal margin matrix $\mathcal{M}$ over $n$ alternatives, there is a unique pair $(a, b_n)$ that is {$\succ$-maximal} such that $a$ is the SSV winner in $\mathcal{M}_{-b_n}$. This is witnessed by a unique pair $(a, b_{n-1})$ that is maximal and such that $a$ is the SSV winner in $(\mathcal{M}_{-b_n})_{-b_{n-1}}$. This continues until we reach $(a, b_2)$. The bijection that sends $a$ to $1$ and $b_i$ to $i$ is therefore a canonical relabeling.  We note that symmetry breaking does not make exhaustively searching among linear ordinal margin matrices tractable; restricting attention to matrices which are not isomorphic to each other, there are still 21,259,022,134,381,903,872,000 matrices for $n = 7$.

\subsection{Results}\label{subsec:results}

For the encoding described above, we apply the Tseitin transformation and run CaDiCal \cite{CaDiCal} release version 1.0.3. The tests were performed on a MacBook Air with an M4 processor and 24 GB of RAM. 

Using the reasoning from the proof of Theorem \ref{prop:proof}, if 1 is SC-defeated, it must be by $n$. Otherwise, there exists a counterexample to the conjecture of smaller size. To search for a minimal counterexample, we add the constraint $SC_{n} \land s_{(n,1),(1,n)}\wedge \lnot r_{(n, 1), n}$ to the formula in Proposition \ref{EncodingProp}.

At $n = 5$, we found a proof of unsatisfiability in 0.01 seconds, verifying Theorem~\ref{prop:proof} given in Section~\ref{5cands}. At $n = 6$, we found a proof of unsatisfiability in $0.5$ seconds, verifying Yifeng Ding's exhaustive search as discussed in Section~\ref{sec:intro}. Thus, we can strengthen Theorem~\ref{prop:proof} as follows.

\begin{theorem}\label{SATstronger} For any linear ordinal margin matrix with at most 6 alternatives, the SSV winner is a Split Cycle winner.
\end{theorem}

However, for $n = 7$, we produced the counterexample to Conjecture \ref{Conj} in Figure~\ref{7candEx} in 73 seconds. The edges are labeled 1 through 21 in order of increasing strength according to $\succ$. Alternative $a$ is the SSV winner and SC-defeated by $b$. The SC winners are $b$, $c$, and $d$. These claims are verified in a notebook in the GitHub repository using the \texttt{pref\_voting} package \citep{HollidayPacuit2025}.

\begin{figure}[h]
\centering
\begin{tikzpicture}[node distance=3cm]
  \node[circle,draw,fill=medgreen!50,minimum width=20] (A) at (0:3) {$a$};
  \node[circle,draw,fill=cyan!50,minimum width=20] (B) at (51.43:3) {$b$};  
  \node[circle,draw,fill=cyan!50,minimum width=20] (C) at (102.86:3) {$c$};
  \node[circle,draw,fill=cyan!50,minimum width=20] (D) at (154.29:3) {$d$};
  \node[circle,draw,minimum width=20] (E) at (205.71:3) {$e$};
  \node[circle,draw,minimum width=20] (F) at (257.14:3) {$f$};
  \node[circle,draw,minimum width=20] (G) at (308.57:3) {$g$};

  \path[->,draw,thick] (D) to[pos=.75]  node[fill=white] {1} (G);
  \path[->,draw,thick] (G) to node[fill=white] {2} (A);
  \path[->,draw,thick] (D) to node[fill=white] {3} (E);
  \path[->,draw,thick] (F) to[pos=.85]  node[fill=white] {4} (D);  
  \path[->,draw,thick] (D) to node[fill=white] {5} (C);
  \path[->,draw,thick] (C) to[pos=.85]  node[fill=white] {6} (E);
  \path[->,draw,thick] (F) to[pos=.85]  node[fill=white] {7} (A);
  \path[->,draw,thick] (B) to[pos=.85]  node[fill=white] {8} (D);
  \path[->,draw,thick] (D) to[pos=.75]  node[fill=white] {9} (A);
  \path[->,draw,thick] (F) to[pos=.75]  node[fill=white] {10} (B);
  \path[->,draw,thick] (A) to[pos=.85]  node[fill=white] {11} (C);
  \path[->,draw,thick] (B) to node[fill=white] {12} (A);
  \path[->,draw,thick] (E) to[pos=.75]  node[fill=white] {13} (B);
  \path[->,draw,thick] (C) to node[fill=white] {14} (B);
  \path[->,draw,thick] (B) to[pos=.85]  node[fill=white] {15} (G);
  \path[->,draw,thick] (G) to[pos=.75]  node[fill=white] {16} (C);
  \path[->,draw,thick] (F) to node[fill=white] {17} (G);
  \path[->,draw,thick] (C) to[pos=.75]  node[fill=white] {18} (F);
  \path[->,draw,thick] (E) to node[fill=white] {19} (F);
  \path[->,draw,thick] (A) to[pos=.75] node[fill=white] {20} (E);  
  \path[->,draw,thick] (G) to[pos=.85]  node[fill=white] {21} (E);
\end{tikzpicture}
\caption{The 7-alternative counterexample to Conjecture \ref{ConjInformal}.}\label{7candEx}
\end{figure}
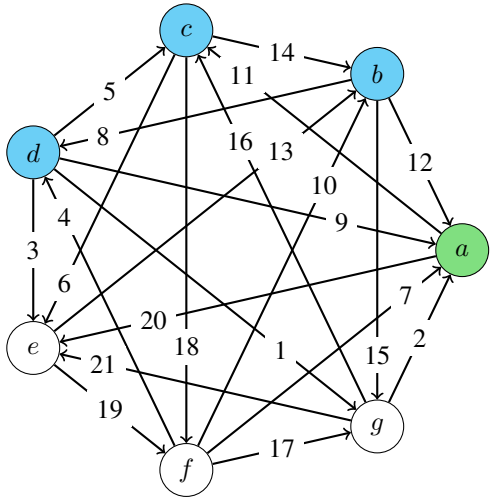

 Using the counterexample in Figure \ref{7candEx}, we show:

\begin{theorem}\label{7andup} For every $n\geq 7$, there is a linear ordinal margin matrix for $n$ alternatives in which the SSV winner is not an SC winner.
\end{theorem}

\begin{proof} Let $k=n-7$. Start with the 7-alternative example in Figure \ref{7candEx} and then expand to an $n$-alternative linear ordinal margin matrix with new alternatives $a_1,\dots,a_k$ such that (i) the descending order of pairs according to $\succ$ begins with $(a,a_k),(a,a_{k-1}),\dots, (a,a_1)$, and (ii) where $c$ is any old alternative, we have $(c,a_i)\succ (a_i,c)$, and the ordering of the pairs of old alternatives remains the same. We prove by induction on $k$ that $a$ is the SSV winner in the expanded ordinal margin matrix. The base case of $k=0$ is given by Figure~\ref{7candEx}. Now suppose the claim holds for $k-1$. To prove it for $k$, by (i) the first pair to consider in computing the SSV winner is $(a,a_k)$, and by the inductive hypothesis, $a$ is the SSV winner after the removal of $a_k$, so $a$ is the SSV winner with all of $a_1,\dots,a_k$ included. Yet $a$ is still SC-defeated with all of $a_1,\dots,a_k$ included, since by (ii) their inclusion does not create any new cycles involving old alternatives.\end{proof}

It is striking that Split Cycle chooses 3 winners for the counterexample in Figure \ref{7candEx}, given that it is estimated to produce only 1.08 winners on average with 7 alternatives in preference profiles sampled according to the Impartial Culture with many voters \cite{HP2023}. Thus, we decided to search for a counterexample with only a single Split Cycle winner. We add constraints $\lnot SC_{a} \lor \lnot SC_{b}$ for each pair of distinct alternatives $a, b \in [n]$. At size 7, this formula was determined unsatisfiable in 6 hours and 45 minutes. Going to size 8, we must remove the constraint $SC_{n} \land \neg r_{(n, 1), n}$, because smaller counterexamples exist, so the reasoning from the beginning of the proof of  Theorem~\ref{prop:proof} no longer applies. With this encoding, we find the counterexample in Figure \ref{8candEx}, where $a$ is the SSV winner and $g$ is the Split Cycle winner, in 24 minutes. In fact, $a$ is SC-defeated by $g$, so it is possible for SSV to choose a winner that is SC-defeated by the SV winner. This graph does not contain a subgraph that has distinct SV and SSV winners.

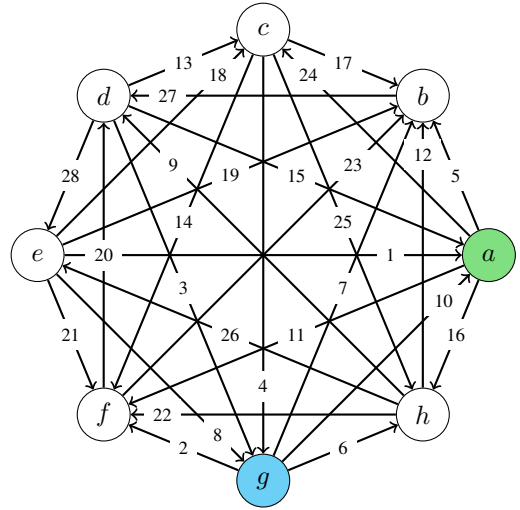
\begin{figure}[h]
\begin{center}
\begin{tikzpicture}[node distance=3cm]
  \node[circle,draw,minimum width=20,fill=medgreen!50] (0) at (0:3) {$a$};
  \node[circle,draw,minimum width=20] (1) at (45:3) {$b$};
  \node[circle,draw,minimum width=20] (2) at (90:3) {$c$};
  \node[circle,draw,minimum width=20] (3) at (135:3) {$d$};
  \node[circle,draw,minimum width=20] (4) at (180:3) {$e$};
  \node[circle,draw,minimum width=20] (5) at (225:3) {$f$};
  \node[circle,draw,minimum width=20,fill=cyan!50] (6) at (270:3) {$g$};
  \node[circle,draw,minimum width=20] (7) at (315:3) {$h$};
  
  \path[->,draw,thick] (0) to[pos=.86] node[fill=white,font=\scriptsize] {24} (2);
  \path[->,draw,thick] (0) to node[fill=white,font=\scriptsize] {16} (7);
  \path[->,draw,thick] (0) to node[fill=white,font=\scriptsize] {11} (5);
  \path[->,draw,thick] (0) to node[fill=white,font=\scriptsize] {5} (1);
  
  \path[->,draw,thick] (1) to[pos=.86] node[fill=white,font=\scriptsize] {27} (3);
  
  \path[->,draw,thick] (2) to node[fill=white,font=\scriptsize] {25} (7);
  \path[->,draw,thick] (2) to node[fill=white,font=\scriptsize] {17} (1);
  \path[->,draw,thick] (2) to node[fill=white,font=\scriptsize] {14} (5);
  \path[->,draw,thick] (2) to[pos=.83] node[fill=white,font=\scriptsize] {4} (6);
  
  \path[->,draw,thick] (3) to node[fill=white,font=\scriptsize] {28} (4);
  \path[->,draw,thick] (3) to node[fill=white,font=\scriptsize] {15} (0);
  \path[->,draw,thick] (3) to node[fill=white,font=\scriptsize] {13} (2);
  \path[->,draw,thick] (3) to node[fill=white,font=\scriptsize] {3} (6);
  
  \path[->,draw,thick] (4) to node[fill=white,font=\scriptsize] {21} (5);
  \path[->,draw,thick] (4) to node[fill=white,font=\scriptsize] {19} (1);
  \path[->,draw,thick] (4) to[pos=.86] node[fill=white,font=\scriptsize] {18} (2);
  \path[->,draw,thick] (4) to[pos=.86] node[fill=white,font=\scriptsize] {8} (6);
  \path[->,draw,thick] (4) to[pos=.82] node[fill=white,font=\scriptsize] {1} (0);
  
  \path[->,draw,thick] (5) to[pos=.82] node[fill=white,font=\scriptsize] {23} (1);
  \path[->,draw,thick] (5) to node[fill=white,font=\scriptsize] {20} (3);
  
  \path[->,draw,thick] (6) to[pos=.86] node[fill=white,font=\scriptsize] {10} (0);
  \path[->,draw,thick] (6) to node[fill=white,font=\scriptsize] {7} (1);
  \path[->,draw,thick] (6) to node[fill=white,font=\scriptsize] {6} (7);
  \path[->,draw,thick] (6) to node[fill=white,font=\scriptsize] {2} (5);
  
  \path[->,draw,thick] (7) to node[fill=white,font=\scriptsize] {26} (4);
  \path[->,draw,thick] (7) to[pos=.88] node[fill=white,font=\scriptsize] {22} (5);
  \path[->,draw,thick] (7) to[pos=.88] node[fill=white,font=\scriptsize] {12} (1);
  \path[->,draw,thick] (7) to[pos=.82] node[fill=white,font=\scriptsize] {9} (3);
\end{tikzpicture}
\end{center}
\caption{The 8-alternative counterexample with one Split Cycle winner ($g$) that defeats the SSV winner ($a$).}\label{8candEx}
\end{figure}

To understand the nature and distribution of counterexamples to Conjecture~\ref{ConjInformal}, we performed a separate counterexample search on each of the 456 tournament\footnote{Recall that a tournament is a directed graph whose edge  relation is an asymmetric and connected binary relation.} isomorphism classes of size 7. To do this, we removed our symmetry-breaking clauses and replaced them with constraints assigning each of the variables $s_{(a, b), (b, a)}$ for distinct $a, b \in [n]$ to match a representative of the given isomorphism class. We found that 115 (25.2\%) of the isomorphism classes contained counterexamples, while the others did not. 

We also performed model enumeration on counterexamples of size 7, where we subsequently disallow the ordering of the edge weights after each counterexample is found. Using this method, we were able to find over 9.8 million models over the course of 19 days with CaDiCal.
 We did not exhaust all models with this search; in fact, progress in finding models remained steady until the end. 

Although our SAT encoding can be formally verified using Lean \citep{HNPZ2024}, formal verification does not seem a pressing issue in this case, since we have a human-readable proof up to $n=5$, a verified DRAT proof for $n=6$, and satisfiability results for $n > 6$, which we verified using the \texttt{pref\_voting} package \citep[\text{https://pref-voting.readthedocs.io/en/latest/}]{HollidayPacuit2025}.

\section{Conclusion and Future Directions}

The negative results of Section~\ref{subsec:results} refute Conjecture~\ref{ConjInformal}. However, in many applications, we seek a voting method for fewer than 7 alternatives. The positive results in Theorems \ref{prop:proof} and \ref{SATstronger} imply that in such cases, the choice between Simple Stable Voting (SSV) and Stable Voting (SV) is immaterial when there are no tied margins: SSV, with its relatively simple definition, can be used to explain the method to voters, while SV, with its superior average-case runtime (see \text{https://github.com/epacuit/stablevoting}), can be used to compute winners more quickly. That said, the computational complexity of computing winners for SV and for SSV remains an open problem.

The techniques employed in this paper are applicable to other questions in voting theory. As an illustration, in our GitHub repository, we show how our SAT encoding can be used to prove that SSV satisfies the property of Reversal Symmetry (\citealt[\S~3.1.3]{Saari1994}, \citealt[\S~7.1]{Saari1997}) up to 6 alternatives: if $a$ is the winner in a linear ordinal margin matrix $(X,\succ)$, then $a$ is not the winner in the reversed ordinal margin matrix $(X,\succ^{-1})$.\footnote{Thinking in terms of preference profiles, Reversal Symmetry requires that if $a$ is the winner in an initial profile, then after all voters reverse their ballots, $a$ is not a winner in the new profile.} We can apply a similar approach to other axioms and to other voting methods that require only the information in the ordinal margin matrix to select winners. As far as we know, there are no complete axiomatic characterizations
of the Ranked Pairs \cite{Tideman1987}, Beat Path \cite{Schulze2011}, or River (\citealt{Doring2025}) methods, or of SV and SSV.\footnote{By contrast, Split Cycle was recently axiomatically characterized \cite{ding2022axiomatic}.} We hope that our approach can be useful for testing conjectured axioms in efforts toward characterizations of these methods.

Contributing to a tradition of work applying SAT solving to computational social choice, the results in this paper provide further evidence of the power of SAT solving to answer open questions that have resisted other techniques.

\appendix

\section{Proof of Lemma \ref{RepThm}}

\begin{proof}
 We first turn $\mathcal{M}=(X,\succ)$ into a matrix $M$  in the usual sense, where the $(a,b)$ entry is the margin of $a$ vs. $b$. Since $\succ$ is a strict weak order, it induces a linear order of equivalence classes under the equivalence relation $\sim$ defined by $(a,b)\sim (c,d)$ if $(a,b)\not\succ (c,d)$ and $(c,d)\not\succ (a,b)$. Starting with the set $S^+ = \{(a,b)\mid (b,a)\not\succ (a,b)\}$ of ``non-negative'' pairs, if there are any pairs $(a,b)\in S^+$ for which $(b,a)\in S^+$, which are then $\succ$-minimal elements in $S^+$, these pairs are assigned weight 0; the pairs in the next weakest equivalence class are assigned weight 2, the next $4$, etc. Let $S^- = \{(a,b)\mid  (b,a)\succ (a,b)\}$ be the set of ``negative'' pairs. For each $(a,b)\in S^-$, the weight assigned to $(a,b)$ is the negative of the weight assigned to $(b,a)$ in the previous step. Now define a relation $\succ_M$ on $X^2$ by $(a,b)\succ_M(c,d)$ iff the weight of $(a,b)$ is greater than that of $(c,d)$ according to $M$. We claim that $(a,b)\succ_M(c,d)$ iff $(a,b)\succ (c,d)$. 
 
 Case 1:  Both of the pairs $(a,b),(c,d)$ are non-negative. Then the claimed equivalence is clear by construction.
 
 Case 2: One pair is non-negative and the other negative. Without loss of generality, suppose $(a,b)\in S^+$ and $(c,d)\in S^-$. Then $(a,b)\succ_M(c,d)$, since the weight of $(a,b)$ is non-negative and the weight of $(c,d)$ is negative in $M$. We claim that $(a,b)\succ (c,d)$ as well. Since $(a,b)\in S^+$ and $(c,d)\in S^-$, we have $(b,a)\not\succ (a,b)$ and $(d,c)\succ (c,d)$. Since $\succ$ is a strict weak order, $(d,c)\succ (c,d)$ implies that either $(a,b)\succ (c,d)$ or $(d,c)\succ (a,b)$. In the first case, we are done. In the second case, by ordinal skew-symmetry, we have $(b,a)\succ (c,d)$. But then since $(b,a)\not\succ (a,b)$ and $\succ$ is a strict weak order, it follows that $(a,b)\succ (c,d)$.

Case 3: $(a,b)$ and $(c,d)$ are both negative.  Then without loss of generality, it suffices to consider two cases.

Case 3a: $(a,b)\succ_M (c,d)$. Then by the construction of $\succ_M$, we have $(d,c)\succ_M(b,a)$. Since $(a,b)$ and $(c,d)$ are negative, $(b,a)$ and $(d,c)$ are positive, 
so $(d,c)\succ(b,a)$ by Case 1, so $(a,b)\succ (c,d)$ by ordinal skew-symmetry of $\succ$. 

Case 3b: $(a,b)\not\succ_M(c,d)$ and $(c,d)\not\succ_M(a,b)$. Then the negative weight assigned to $(a,b)$ by $M$ is the same as that assigned to $(c,d)$ by $M$, which implies that the positive weight assigned to $(b,a)$ by $M$ is the same as that assigned to $(d,c)$ by $M$, which implies $(b,a)\not\succ (d,c)$ and $(d,c)\not\succ (b,a)$ by construction of $M$, which in turn implies  $(c,d)\not\succ (a,b)$ and $(a,b)\not\succ (c,d)$ by  ordinal skew-symmetry.

Now clearly $M$ is a skew-symmetric matrix with even integer entries, so by Debord's Theorem \cite{Debord1987} (see \cite[Theorem 4.1]{Fischer2016}), there is a preference profile $\mathbf{P}$ whose margin matrix is~$M$. Then  $\mathcal{M}(\mathbf{P})=(X,\succ_M) = (X,\succ)=\mathcal{M}$.\end{proof}

\section{Tied Margins}\label{Ties}

In this appendix, we address the status of Conjecture \ref{Conj} if we drop the restriction to \textit{linear} ordinal margin matrices. First, if we drop this restriction, we must decide how to define Simple Stable Voting in the presence of tied margins. The simplest approach, adopted in \cite[Footnote~7]{HP2023}, is the following.

\begin{definition} For any ordinal margin matrix $\mathcal{M}$, the winners according to Simple Stable Voting with Simultaneous Elimination (SSV-SE) are defined recursively as follows:
\begin{enumerate}
\item if $\mathcal{M}$ has only one alternative, this alternative is the SSV-SE winner in $\mathcal{M}$;
\item if $\mathcal{M}$ has more than one alternative, then the SSV-SE winners in $\mathcal{M}$ are the first coordinates of the $\succ$-maximal elements of  $\{(a,b)\mid \mbox{$a$ is an SSV-SE winner in $\mathcal{M}_{-b}$}\}$.
\end{enumerate} 
\end{definition}

Interestingly, while Conjecture \ref{Conj} holds for linear ordinal margin matrices up to 6 alternatives (Theorem~\ref{SATstronger}), if we allow tied margins, then the conjecture only holds up to 4.

\begin{proposition} For any ordinal margin matrix $\mathcal{M}$ with $\leq 4$ alternatives, the SSV-SE winners are among the SC winners.
\end{proposition}

\begin{proof} 
For two alternatives the result is immediate, and for three alternatives it follows from Proposition 1 of \cite{HP2023b}, which characterizes the SSV-SE winners in three-alternative elections. 

Finally, suppose $\mathcal{M}$ has four alternatives. Suppose for a contradiction that $(a,b)$ is $\succ$-maximal such that $a$ is an SSV-SE winner in $\mathcal{M}_{-b}$, but $a$ is not an SC winner in $\mathcal{M}$. Since $a$ is an SSV-SE winner in $\mathcal{M}_{-b}$, which has three alternatives, $a$ is an SC winner in this submatrix; it follows that $b$ (and only $b$) SC-defeats $a$ in $\mathcal{M}$. Now let $c$ be an SSV-SE winner in $\mathcal{M}_{-a}$, and let $d$ be an SSV-SE winner in $\mathcal{M}_{-c}$. Given that $a$ has a negative margin vs. $b$, i.e., $(b,a)\succ (a,b)$, the previous sentence implies $c$ has a negative margin vs. $a$ and $d$ has a negative margin vs. $c$, for otherwise $(a,b)$ would not be $\succ$-maximal such that $a$ is an SSV-SE winner in $\mathcal{M}_{-b}$. If $c = b$, then $a$ is not an SSV-SE winner in $\mathcal{M}$ after all, since $(b,a) \succ (a,b)$, and $b$ ($=c$) wins when $a$ is removed. If $d=a$, then $a$ is an SC winner in $\mathcal{M}_{-c}$, which contains $b$, so $b$ does not SC-defeat $a$ after all. If $d=b$, then since we know $(c,a), (d,c)\not\succ (a,b)$, the edge $(b,a)$ is a weakest edge in the cycle $(b,a,c)$, so $b$ does not SC-defeat $a$ after all. Thus, we have distinct alternatives $b,a,c,d$. Let $e$ be an SSV-SE winner in $\mathcal{M}_{-d}$. By reasoning analogous to that proving $d\neq a,b$, we have $e \neq a,b$. Hence $e=c$ since $|X|=4$. But then $a$ is not an SSV-SE winner in $\mathcal{M}$ after all, since $(c,d) \succ (a,b)$, and $c$ ($=e$) is an SSV-SE winner in $\mathcal{M}_{-d}$.
\end{proof}

\begin{example} Yifeng Ding (personal communication) provided the example in Figure~\ref{YifengEx} in which SSV-SE does not refine SC.

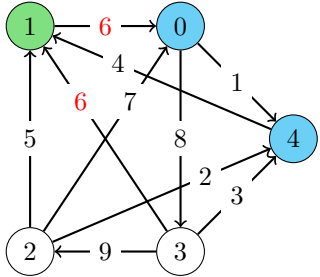
\begin{figure}[h]
\begin{center}
\begin{tikzpicture}
\node[circle,draw,minimum width=0.25in,fill=cyan!50] at (2,1.5)  (a) {$0$}; 
\node[circle,draw,minimum width=0.25in,fill=medgreen!50] at (0,1.5)  (b) {$1$}; 
\node[circle,draw,minimum width=0.25in] at (0,-1.5) (c) {$2$}; 
\node[circle,draw,minimum width=0.25in] at (2,-1.5) (d) {$3$}; 
\node[circle,draw,minimum width=0.25in,fill=cyan!50] at (3.5,0)  (e) {$4$};
\path[->,draw,thick] (b) to node[fill=white] {$\textcolor{red}{6}$} (a);
\path[->,draw,thick] (a) to node[fill=white] {$8$} (d);
\path[->,draw,thick] (a) to node[fill=white] {$1$} (e);
\path[->,draw,thick] (c) to node[fill=white] {$5$} (b);
\path[->,draw,thick] (d) to node[fill=white] {$9$} (c);
\path[->,draw,thick] (d) to node[fill=white] {$3$} (e);
\path[->,draw,thick] (c) to[pos=.7] node[fill=white] {$7$} (a);
\path[->,draw,thick] (d) to[pos=.7] node[fill=white] {$\textcolor{red}{6}$} (b);
\path[->,draw,thick] (e) to[pos=.7] node[fill=white] {$4$} (b);
\path[->,draw,thick] (c) to[pos=.7] node[fill=white] {$2$} (e);
\end{tikzpicture}
\end{center}
\caption{An example with tied margins (red) in which the SSV-SE winner (green) is not among the SC-winners (blue).}\label{YifengEx}
\end{figure}
\end{example}

However, SSV-SE is not the only way of generalizing SSV to allow tied margins. We will now consider an approach to tied margins inspired by the parallel universe tie-breaking used for the Ranked Pairs rule \cite{Tideman1987,BrillFischer2012}: faced with tied margins, we consider all possible ways of breaking such ties, and if an alternative wins under one such way, then they count as a winner.

\begin{definition} Given an ordinal margin matrix ${\mathcal{M}=(X,\succ)}$, we say that a linear ordinal margin matrix $\mathcal{M}'=(X,\succ')$ is a \textit{linearization} of $\mathcal{M}$ if for all $a,b,c,d\in X$, $(a,b)\succ (c,d)$ implies $(a,b)\succ' (c,d)$.
\end{definition}

\begin{definition}\label{SSVPUTdef} For any ordinal margin matrix $\mathcal{M}$, the winners according to Simple Stable Voting with Parallel Universe Tiebreaking (SSV-PUT) are defined by 
\[SSV\mbox{-}PUT(\mathcal{M}) = \underset{\mathcal{M}'\mbox{ \footnotesize a linearization of }\mathcal{M}}{\bigcup}SSV(\mathcal{M}').\]
\end{definition}

\noindent To relate SSV-PUT to Split Cycle, we use the following fact about how Split Cycle behaves with respect to linearizations.

\begin{lemma}\label{SClinearization} For any ordinal margin matrix $\mathcal{M}$ and linearization $\mathcal{M}'$ of $\mathcal{M}$, we have $SC(\mathcal{M}')\subseteq SC(\mathcal{M})$.\end{lemma}

\begin{proof} Let $\mathcal{M}=(X,\succ)$ and $\mathcal{M}'=(X,\succ')$. Suppose $a\in SC(\mathcal{M}')$. To show that $a\in SC(\mathcal{M})$, consider any $b\in X$ such that $(b,a)\succ(a,b)$, so $(b,a)$ is an edge in the majority graph $G(X,\succ)$. We must show that there is a cycle in $G(X,\succ)$ such that $(b,a)$ is an edge that is weakest in the cycle according to $\succ$. Since $(b,a)\succ(a,b)$ and $\mathcal{M}'$ is a linearization of $\mathcal{M}$, we have $(b,a)\succ'(a,b)$. But then since $a\in SC(\mathcal{M}')$, there is a cycle $\rho$ in the majority graph $G(X,\succ')$ such that $(b,a)$ is an edge that is weakest in the cycle according to $\succ'$.  We claim that $\rho$ is a cycle in $G(X,\succ)$ such that $(b,a)$ is weakest in the cycle according to $\succ$. First, we show that each edge $(x,y)$ in $\rho$ is an edge in $G(X,\succ)$, i.e., $(x,y)\succ (y,x)$. For if $(x,y)\not\succ (y,x)$, then since $(b,a)\succ (a,b)$, it follows by the properties of $\succ$ that $(b,a)\succ (x,y)$, in which case $(b,a)\succ' (x,y)$ since $\mathcal{M}'$ is a linearization of $\mathcal{M}$, contradicting the fact that $(b,a)$ is the weakest edge in $\rho$ according to $\succ'$. It follows that $\rho$ is also a cycle in $G(X,\succ)$. Moreover, $(b,a)$ is a weakest edge in $\rho$ according to $\succ$, for if there is another edge $(w,v)$ in $\rho$ such that $(b,a)\succ (w,v)$, then $(b,a)\succ'(w,v)$ since $\mathcal{M}'$ is a linearization of $\mathcal{M}$, contradicting the fact that $(b,a)$ is the weakest edge in $\rho$ according to $\succ$. Thus, we have shown that for any $b\in X$ such that $(b,a)\succ (a,b)$, there is a cycle in $G(X,\succ)$ such that $(b,a)$ is a weakest edge in the cycle according to $\succ$, which establishes that $a\in SC(\mathcal{M})$.\end{proof}

Now we can show that SSV-PUT is ``as good'' at refining Split Cycle on arbitrary ordinal margin matrices as SSV is at refining Split Cycle on linear ordinal margin matrices.

\begin{proposition}\label{SSVPUTinclusion} Let $\mathcal{M}$ be an ordinal margin matrix. If for all linearizations $\mathcal{M}'$ of $\mathcal{M}$, we have $SSV(\mathcal{M}')\subseteq SC(\mathcal{M}')$, then $SSV\mbox{-}PUT(\mathcal{M})\subseteq SC(\mathcal{M})$.
\end{proposition}
\begin{proof}We have 
\begin{eqnarray*}
SSV\mbox{-}PUT(\mathcal{M}) &=& \underset{\mathcal{M}'\mbox{ \footnotesize a linearization of }\mathcal{M}}{\bigcup}SSV(\mathcal{M}') \\
&\subseteq&  \underset{\mathcal{M}'\mbox{ \footnotesize a linearization of }\mathcal{M}}{\bigcup}SC(\mathcal{M}')\\
&\subseteq& SC(\mathcal{M})\end{eqnarray*}
using Lemma \ref{SClinearization} for the last inclusion.\end{proof}

Combining Theorem \ref{SATstronger} and Proposition \ref{SSVPUTinclusion}, we have the following immediate corollary. 

\begin{corollary} For all ordinal margin matrices $\mathcal{M}$ with 6 or fewer alternatives, $SSV\mbox{-}PUT(\mathcal{M})\subseteq SC(\mathcal{M})$.
\end{corollary}
\noindent On the other hand, since $SSV\mbox{-}PUT(\mathcal{M})=SSV(\mathcal{M})$ for any linear ordinal margin matrix $\mathcal{M}$, the example in Section~\ref{subsec:results} shows that SSV-PUT does not refine Split Cycle in all $7$-alternative elections.

\bibliography{refs}

\begin{thebibliography}{32}
\providecommand{\natexlab}[1]{#1}

\bibitem[{Biere(2019)}]{CaDiCal}
Biere, A. 2019.
\newblock {CaDiCaL} at the {SAT} {Race} 2019.
\newblock In Heule, M.; J{\"a}rvisalo, M.; and Suda, M., eds., \emph{Proc.~of
  {SAT Race} 2019 -- Solver and Benchmark Descriptions}, volume B-2019-1 of
  \emph{Department of Computer Science Series of Publications B}, 8--9.
  University of Helsinki.

\bibitem[{Brandt et~al.(2016)Brandt, Conitzer, Endriss, Lang, and
  Procaccia}]{Brandt2016Handbook}
Brandt, F.; Conitzer, V.; Endriss, U.; Lang, J.; and Procaccia, A.~D., eds.
  2016.
\newblock \emph{Handbook of Computational Social Choice}.
\newblock Cambridge: Cambridge University Press.

\bibitem[{Brandt and Geist(2016)}]{Brandt2016}
Brandt, F.; and Geist, C. 2016.
\newblock Finding Strategyproof Social Choice Functions via {SAT} Solving.
\newblock \emph{Journal of Artificial Intelligence Research}, 55: 565--602.

\bibitem[{Brandt, Geist, and Peters(2017)}]{Brandt2017}
Brandt, F.; Geist, C.; and Peters, D. 2017.
\newblock Optimal bounds for the no-show paradox via {SAT} solving.
\newblock \emph{Mathematical Social Sciences}, 90: 18--27.

\bibitem[{Brandt, Geist, and Seedig(2014)}]{Brandt2014}
Brandt, F.; Geist, C.; and Seedig, H.~G. 2014.
\newblock Identifying k-Majority Digraphs via {SAT} Solving.
\newblock In \emph{Proceedings of the 1st AAMAS Workshop on Exploring Beyond
  the Worst Case in Computational Social Choice}.

\bibitem[{Brandt, Saile, and Stricker(2018)}]{Brandt2018}
Brandt, F.; Saile, C.; and Stricker, C. 2018.
\newblock Voting with Ties: Strong Impossibilities via {SAT} Solving.
\newblock In Dastani, M.; Sukthankar, G.; Andr\'{e}, E.; and Koenig, S., eds.,
  \emph{Proceedings of the 17th International Conference on Autonomous Agents
  and Multiagent Systems (AAMAS 2018)}, 1285--1293. International Foundation
  for Autonomous Agents and Multiagent Systems.

\bibitem[{Brandt, Saile, and Stricker(2022)}]{Brandt2022}
Brandt, F.; Saile, C.; and Stricker, C. 2022.
\newblock Strategyproof social choice when preferences and outcomes may contain
  ties.
\newblock \emph{Journal of Economic Theory}, 202: 105447.

\bibitem[{Brill and Fischer(2012)}]{BrillFischer2012}
Brill, M.; and Fischer, F. 2012.
\newblock The Price of Neutrality for the Ranked Pairs Method.
\newblock In \emph{Proceedings of the Twenty-Sixth AAAI Conference on
  Artificial Intelligence (AAAI-12)}, 1299--1305. AAAI Press.

\bibitem[{Condorcet(1785)}]{Condorcet1785}
Condorcet, M. 1785.
\newblock \emph{Essai sur l'application de l'analyse \`{a} la probabiliti\'{e}
  des d\'{e}cisions rendues \`{a} la pluralit\'{e} des voix}.
\newblock Paris: l'Imprimerie Royale.

\bibitem[{Debord(1987)}]{Debord1987}
Debord, B. 1987.
\newblock Caract\'{e}risation des matrices des pr\'{e}f\'{e}rences nettes et
  m\'{e}thodes d'agr\'{e}gation associ\'{e}es.
\newblock \emph{Math\'{e}matiques et sciences humaines}, 97: 5--17.

\bibitem[{Ding, Holliday, and Pacuit(2025)}]{ding2022axiomatic}
Ding, Y.; Holliday, W.~H.; and Pacuit, E. 2025.
\newblock An Axiomatic Characterization of {S}plit {C}ycle.
\newblock \emph{Social Choice and Welfare}, 64: 557--601.

\bibitem[{D{\"o}ring, Brill, and Heitzig(2025)}]{Doring2025}
D{\"o}ring, M.; Brill, M.; and Heitzig, J. 2025.
\newblock The River Method.
\newblock \emph{arXiv preprint arXiv:2504.14195}.

\bibitem[{Dutta and Laslier(1999)}]{Dutta1999}
Dutta, B.; and Laslier, J.-F. 1999.
\newblock Comparison functions and choice correspondences.
\newblock \emph{Social Choice and Welfare}, 16: 513--532.

\bibitem[{Enderton(2001)}]{Enderton2001}
Enderton, H.~B. 2001.
\newblock \emph{A Mathematical Introduction to Logic}.
\newblock Harcourt Academic Press.

\bibitem[{Fischer, Hudry, and Niedermeier(2016)}]{Fischer2016}
Fischer, F.; Hudry, O.; and Niedermeier, R. 2016.
\newblock Weighted Tournament Solutions.
\newblock In Brandt, F.; Conitzer, V.; Endriss, U.; Lang, J.; and Procaccia,
  A.~D., eds., \emph{Handbook of Computational Social Choice}, 85--102. New
  York: Cambridge University Press.

\bibitem[{Fishburn(1977)}]{Fishburn1977}
Fishburn, P.~C. 1977.
\newblock Condorcet Social Choice Functions.
\newblock \emph{SIAM Journal on Applied Mathematics}, 33(3): 469--489.

\bibitem[{Holliday et~al.(2024)Holliday, Norman, Pacuit, and
  Zahedian}]{HNPZ2024}
Holliday, W.~H.; Norman, C.; Pacuit, E.; and Zahedian, S. 2024.
\newblock Impossibility theorems involving weakenings of expansion consistency
  and resoluteness in voting.
\newblock In Jones, M.~A.; McCune, D.; and Wilson, J.~M., eds.,
  \emph{Mathematical Analyses of Decisions, Voting and Games}, volume 795 of
  \emph{Contemporary Mathematics}. Ann Arbor, MI: American Mathematical
  Society.

\bibitem[{Holliday and Pacuit(2023{\natexlab{a}})}]{HP2023}
Holliday, W.~H.; and Pacuit, E. 2023{\natexlab{a}}.
\newblock Split {C}ycle: a new {C}ondorcet-consistent voting method independent
  of clones and immune to spoilers.
\newblock \emph{Public Choice}, 197: 1--62.

\bibitem[{Holliday and Pacuit(2023{\natexlab{b}})}]{HP2023b}
Holliday, W.~H.; and Pacuit, E. 2023{\natexlab{b}}.
\newblock Stable {V}oting.
\newblock \emph{Constitutional Political Economy}, 34: 421–433.

\bibitem[{Holliday and Pacuit(2024)}]{HP2024data}
Holliday, W.~H.; and Pacuit, E. 2024.
\newblock Stable {V}oting {D}ataset, {R}elease: 12-17-2024.

\bibitem[{Holliday and Pacuit(2025)}]{HollidayPacuit2025}
Holliday, W.~H.; and Pacuit, E. 2025.
\newblock pref\_voting: The Preferential Voting Tools package for Python.
\newblock \emph{Journal of Open Source Software}, 10(105): 7020.

\bibitem[{Kluiving et~al.(2020)Kluiving, de~Vries, Vrijbergen, Boixel, and
  Endriss}]{Kluiving2020}
Kluiving, B.; de~Vries, A.; Vrijbergen, P.; Boixel, A.; and Endriss, U. 2020.
\newblock Analysing Irresolute Multiwinner Voting Rules with Approval Ballots
  via {SAT} Solving.
\newblock In \emph{Proceedings of the 24th European Conference on Artificial
  Intelligence (ECAI-2020)}, 131--138.

\bibitem[{Kramer(1977)}]{Kramer1977}
Kramer, G.~H. 1977.
\newblock A dynamical model of political equilibrium.
\newblock \emph{Journal of Economic Theory}, 16(2): 310--334.

\bibitem[{P\'{e}rez-Fern\'{a}ndez and {De Baets}(2018)}]{Fernandez2018}
P\'{e}rez-Fern\'{a}ndez, R.; and {De Baets}, B. 2018.
\newblock The supercovering relation, the pairwise winner, and more missing
  links between {B}orda and {C}ondorcet.
\newblock \emph{Social Choice and Welfare}, 50: 329--352.

\bibitem[{Saari(1994)}]{Saari1994}
Saari, D.~G. 1994.
\newblock \emph{Geometry of Voting}.
\newblock Berlin: Springer.

\bibitem[{Saari(1999)}]{Saari1997}
Saari, D.~G. 1999.
\newblock Explaining All Three-Alternative Voting Outcomes.
\newblock \emph{Journal of Economic Theory}, 87(2): 313--355.

\bibitem[{Schulze(2011)}]{Schulze2011}
Schulze, M. 2011.
\newblock A new monotonic, clone-independent, reversal symmetric, and
  condorcet-consistent single-winner election method.
\newblock \emph{Social Choice and Welfare}, 36: 267--303.

\bibitem[{Simpson(1969)}]{Simpson1969}
Simpson, P.~B. 1969.
\newblock On Defining Areas of Voter Choice: {P}rofessor {T}ullock on Stable
  Voting.
\newblock \emph{The Quarterly Journal of Economics}, 83(3): 478--490.

\bibitem[{Tideman(1987)}]{Tideman1987}
Tideman, T.~N. 1987.
\newblock Independence of Clones as a Criterion for Voting Rules.
\newblock \emph{Social Choice and Welfare}, 4: 185--206.

\bibitem[{Tseitin(1983)}]{Tseitin}
Tseitin, G.~S. 1983.
\newblock On the Complexity of Derivation in Propositional Calculus.
\newblock In Siekmann, J.~H.; and Wrightson, G., eds., \emph{Automation of
  Reasoning: 2: Classical Papers on Computational Logic 1967--1970}, 466--483.
  Berlin, Heidelberg: Springer Berlin Heidelberg.

\bibitem[{Young(1988)}]{Young1988}
Young, H.~P. 1988.
\newblock Condorcet's Theory of Voting.
\newblock \emph{American Political Science Review}, 82(4): 1231--1244.

\bibitem[{Zwicker(2016)}]{Zwicker2016}
Zwicker, W.~S. 2016.
\newblock Introduction to the Theory of Voting.
\newblock In Brandt, F.; Conitzer, V.; Endriss, U.; Lang, J.; and Procaccia,
  A.~D., eds., \emph{Handbook of Computational Social Choice}, 23--56. New
  York: Cambridge University Press.

\end{thebibliography}

\end{document}